\documentclass[journal]{IEEEtran}

\usepackage{pgf}
\usepackage{subcaption}
\usepackage{tikz}
\usetikzlibrary{arrows,automata,positioning,shapes}
\usetikzlibrary{decorations.pathmorphing}
\tikzset{snake it/.style={decorate, decoration=snake}}\usepackage{xxcolor}
\usepackage{mathtools}
\usepackage[normalem]{ulem}
\mathtoolsset{showonlyrefs}
\usepackage{theoremref}
\usetikzlibrary{arrows,shadows,petri}

\usepackage{amsfonts,amsmath,amsthm,amssymb}
\usepackage{graphicx,hyperref}

\usepackage{ifthen}
\usepackage{color}
\usepackage{cite}

\newtheorem{theorem}{Theorem}
\newtheorem{lemma}{Lemma}
\newtheorem{proposition}{Proposition}
\newtheorem{corollary}{Corollary}
\newtheorem{assumption}{Assumption}
\newtheorem{remark}{Remark}

\newcommand{\0}{\mathbf{0}}

\newtheorem{definition}{Definition}

\def\mcal{\mathcal}

\title{Going Viral:\\ 
Stability of Consensus-Driven Adoptive Spread}
\author{Sebastian F. Ruf,  Keith Paarporn, Philip E. Par\'{e}$^*$\thanks{$^*$
Sebastian F. Ruf is with Center for Complex Networks Research at Northeastern University and can be reached at {\tt s.ruf@northeastern.edu}.   Keith Paarporn is with the Department of Electrical and Computer Engineering at University of California Santa Barbara, {\tt kpaarporn@ucsb.edu}. Philip E. Par\'{e} is with the Division of Decision and Control Systems at KTH (\texttt{philipar@kth.se}).  This material is based on research partially sponsored by the National Science Foundation, grants CPS 1544953 and ECCS 1509302.
}}

\begin{document}
\maketitle
\begin{abstract}
The spread of new products in a networked population is often modeled as an epidemic. However, in the case of ``complex'' contagion, these models {do not capture nuanced, dynamic social reinforcement effects in adoption behavior}. In this paper, we investigate a model of complex contagion which allows a coevolutionary interplay between adoption, modeled as an SIS epidemic spreading process, and social reinforcement effects, modeled as consensus opinion dynamics. Asymptotic stability analysis of the all-adopt as well as the none-adopt equilibria of the combined opinion-adoption model is provided through the use of Lyapunov arguments. In doing so, sufficient conditions are provided which determine the stability of the ``flop'' state, where no one adopts the product and everyone's opinion of the product is least favorable, and the ``hit'' state, where everyone adopts and their opinions are most favorable. These conditions are shown to extend to the bounded confidence opinion dynamic under a stronger assumption on the model parameters. To conclude, numerical simulations demonstrate behavior of the model which reflect findings from the sociology literature on adoption behavior.
\end{abstract}
\section{Introduction}
    
How technologies, behaviors, and ideas (which we refer to generally as innovations or products) spread is a central question in the study of human behavior. The spread of innovations has been studied since the work of Tarde \cite{tarde1903laws} and is currently an area of active research, see for example \cite{valentehistory,centolabook,kiesling2012agent}. Understanding the complex spreading process of innovations has important ramifications, including the economic impact of understanding how products spread \cite{kiesling2012agent,rogers2010diffusion} as well as the potential benefits that could come from spreading healthy behaviors \cite{centolabook,Valente49,christakis2008collective,christakis2007spread}.  
    
    Innovation diffusion processes are often modeled either with epidemic models, where the innovation is a virus \cite{VM2009,Zhang_2014}, or with ``epidemic-like'' models, in which those that have a product will spread it to others { after a single contact} with some probability \cite{bass1969new,valente1996network,kalish1985product,reviewdiffusion}. This conceptual linkage between products and viruses has found its way into popular culture, resulting in companies hoping that their products will go ``viral'' and spread throughout the population. However, it has been shown that there are types of spread, specifically those that require social reinforcement { or multiple contacts}, which are not well captured by these ``simple'' epidemic models \cite{Centola1194,centolabook,guilbeault2018complex}. { For example, knowledge of a job opening spreads more widely through a network of acquaintances \cite{granweak} while unproven or risky technology would require social validation before adoption \cite{centola2007complex}.} As many behavioral diffusion processes are known to depend on social interaction \cite{rogers2010diffusion,christakis2008collective,christakis2007spread}, this has led to the study of ``complex'' contagions \cite{centolabook}, which take these peer effects into account. 

    Complex contagions are commonly studied using threshold models \cite{valente1996network,schelling1971dynamic,Centola1194,granovetter1978threshold} { or conceptual extensions of the epidemic models such as the independent cascade model \cite{Kempe_2003,shakarian2015independent}}. In a threshold model, an agent adopts if a given number \cite{granovetter1978threshold} or fraction \cite{watts2002simple} of their neighbors has adopted the innovation, which captures social reinforcement effects.

    { This paper, following the approach of \cite{ruf2017opinion}, modifies an epidemic adoption model to allow an opinion dynamic, which captures social reinforcement effects, to affect the spread parameters of the epidemic model. By doing so, a ``simple'' epidemic model is modified to capture ``complex'' adoption behavior. Explicitly modeling the dynamics of social reinforcement allows their effects to change over time and illustrates how different models of social reinforcement, i.e. varying opinion dynamics, influence the spread of a product.}  
    
The coupling between an epidemic model and an opinion dynamic draws a link between the literature on epidemic models \cite{kermack1932contributions,FallMMNP07,VM2009,khanafer2016stability,pare2017epidemic,nowzari2016analysis,Funk2009,human2,Arenas_2013,Paarporn_tcss,liu2017human}. 
and the literature on opinion dynamic models \cite{proskurnikov2018tutorial,abelson1964mathematical,degroot1974reaching,olfati2004consensus,ren2005consensus,olfati2007consensus,nedich2015convergence}. The choice of a single virus epidemic model also means that the presented model extends the literature on the diffusion of innovations \cite{ryan1943diffusion, bass1969new, rogers1971communication,rogers2010diffusion, valente1996network,kiesling2012agent}, which considers the spread of a new product that does not yet have any competitors.

{ 
State-dependent epidemic spreading has received growing attention in recent years \cite{nowzari2016analysis}. For example, there is interest in studying spreading processes when coupled with human awareness, in which awareness and information hinders disease spread \cite{Funk2009,human2,Arenas_2013,Paarporn_tcss,liu2017human}. The primary contributions of these studies illustrate, through equilibrium stability analyses, how awareness and social distancing can mitigate epidemic spreading. A major difference in the proposed model from the above studies is that in the proposed model opinions can promote as well as hinder the spread of innovation.
}

{
There is also a conceptually similar notion of the bi-virus model which studies two competing SIS viruses \cite{nowak1991evolution,prakash2012winner,wei2013competing,sahneh2014competitive,liu2016onthe,arxiv,watkins2016optimal}.
In \cite{nowak1991evolution}, Nowak introduced the idea of competing SIS models with no graph structure but three groups: 1) susceptible, 2) infected with virus one, and 3) infected with virus two. 
 
Watkins {\em et al.}, provide a necessary and sufficient condition for local exponential stability of the origin for two competing heterogeneous viruses over strongly connected graphs in \cite{watkins2016optimal}. A geometric program is also formulated to control the spread of the virus.
In \cite{liu2016onthe,arxiv}, Liu {\em et al.} provide global analysis for the healthy and epidemic states for the bi-virus model over strongly connected graphs and investigate distributed control techniques.
These models are related to the model proposed here because they are layered networks, they modify the spread of the virus, and they allow for more complex behavior, such as a possible spectrum of endemic equilibria \cite{liu2016stability}.

However, the model presented in this paper is distinct in the coupling between the process dynamics as well as the underlying dynamics associated with each layer. 
}
    
   This paper considers the susceptible-infected-susceptible (SIS) epidemic model as the underlying adoption mechanism \cite{FallMMNP07,VM2009,khanafer2016stability,pare2017epidemic,nowzari2016analysis,Zhang_2014}. { We select the SIS compartmental model over others, e.g. SIR \cite{Fibich_2016,shakarian2015sir}, because we interpret an agent as being able to alternate between adopter and non-adopter over time. In an SIR model, once the agent stops using a technology, they can no longer start again.}
  
   For example, an agent could adopt a product, then dislike and therefore discard said product. However, if their network neighbors begin to use and like the product, their opinion might change to the point that they re-adopt. 
    { The evolution of individual opinions is often modeled as consensus opinion dynamics,} which originated in mathematical sociology \cite{abelson1964mathematical,degroot1974reaching} and has since been explored extensively in the controls literature \cite{olfati2004consensus,ren2005consensus,olfati2007consensus,nedich2015convergence,proskurnikov2018tutorial,nedic2008constrained}.
      Consensus dynamics describe information exchange in human and robotic systems \cite{mesbahi2010graph} and have been used to predict opinion change in online experiments\cite{Friedkin11380,BeckerE5070}. In this paper we focus on understanding the impact of consensus opinion dynamics on epidemic spreading. We also extend our consideration to the continuous time bounded confidence opinion model \cite{hegselmann2002opinion,canuto2008eulerian,blondel2010continuous}, which allows links to be severed in the opinion graph.  
    
    There are few existing models that combine adoption with opinion dynamics. 
    In \cite{kalish1985product} a coupled adoption and awareness model is proposed that includes advertising. Similar to \cite{kermack1932contributions}, this model assumes no graph structure, and models the system with only two differential equations, aggregating the population into two groups, susceptible (non-adopters) or infected (adopters). 
Introduced in \cite{Martins2009}, the Continuous Opinion Discrete Action (CODA) model captures discrete product adoption with Bayesian opinion updates, which do not depend on their, or their neighbors', opinions but only on the adoption actions of their network neighbors. 
      
    The specific contributions of this paper include:
    \begin{itemize}
    \item We extend the results of the conference paper \cite{ruf2017opinion} 
    { 
    by modifying the model to include a sensible scaling factor and analyzing the global behavior
    }of the proposed model, providing  conditions for the global stability of the { none-adopt equilibrium in the case of a desirable product and the all-adopt equilibrium in the case of an exceptional product. } 
    We extend this characterization to the case of the bounded confidence opinion dynamic. 
    \item We characterize a class of unstable equilibria { in the case of an exceptional product.} We show via simulation that these equilibria (when they exist in the interior of the state space) cause the evolution of the system to become initial-condition dependent.  
    \end{itemize}
{ These stability results lay a foundation to understanding complex opinion-adoption dynamics through viral models. These insights provide future opportunities to examine questions about control - what aspects of this dynamical process should be influenced, and how, in order to achieve desirable global behaviors.  }
\vspace{-3ex}

\subsection{Notation}

For a vector $x(t)$, $\dot{x}$ indicates its time derivative. We use $1_N$ and $0_N$ to indicate vectors in $\mathbb{R}^N$ of all ones and zeros, respectively. The norm operator $\|\cdot \|$ is the Euclidean or 2-norm. For a matrix $A$, $\sigma(A)$ is the set of eigenvalues of $A$ and $\alpha(A) := \max \left\{{\rm Re}(\lambda)\ : \ \lambda\in\sigma(A)\right\}$. A diagonal matrix with its $ii$th entry being $x_i$
is denoted by $\mathrm{diag}\left(x_{i}\right)$. We define $(0,1)^{N}$ as the $N$ Cartesian products of the interval $(0,1)$ and $[0,1]^{N}$ as the $N$ Cartesian products of the interval~$[0,1]$.

\vspace{-1ex}

\section{Model}\label{sec:model}

We modify the standard network dependent SIS epidemic ODE dynamics to incorporate the coupling between the ``epidemic-like'' spread of an innovation and opinion update dynamics. The adoption dynamics occur over a weighted, directed network $\mcal{G}_A$ of $N$ nodes, or subpopulations. The opinion dynamics occur over a weighted digraph $\mcal{G}_O$ with the same node set as $\mcal{G}_A$, but whose edges may or may not coincide with $\mcal{G}_A$. This captures that for a given node, the other nodes with which the node discusses the product can be distinct from the group of nodes which are observed using the product. The neighborhood set of node $i$ is denoted as $\mcal{N}^X_i$ for $X=A,O$. 

Each node, or subpopulation, $i$ has a proportion of agents that have adopted the product $x_i\in[0,1]$. The subpopulation represented by node $i$ also has an overall average opinion $o_i\in[0,1]$, modeling how much the subpopulation values the product ($o_i=0$ means the subpopulation is averse to the product, $o_i=1$ means very receptive to the product). 

 Networked epidemic models have two different interpretations: the node state can correspond to 1) the probability of an individual being infected \cite{VM2009}, or 2) the proportion of a subpopulation that is infected \cite{FallMMNP07,bivirus_tac}. In this paper we use the latter interpretation since adoption is a binary action, making this work distinct from the body of work that explores modeling adoption as a discrete process \cite{granovetter1978threshold,valente1996network,Martins2009}.
 
 The adoption dynamics for each node evolve as a function of time:
 \vspace{-1ex}
    \begin{align}
    \dot{x}_i &= f_i(x,o) \nonumber \\
    &\equiv -\delta_i x_i (1-o_i) + (1-x_i)o_i\left( \sum_{\mcal{N}_i^A}\beta_{ij} x_j +\beta_{ii} \right) \label{eq:prod} 
    \end{align}
where $\delta_i > 0$ is the drop rate for subpopulation $i$, $\beta_{ij} \geq 0$ is the exogenous adoption rate, and $\beta_{ii}\geq 0$ is the endogenous adoption rate. The parameters $\beta_{ij}$ are the weights on the adoption graph. This model captures that a subpopulation's opinion towards a product directly impacts their proclivity to adopt or to stop using a product. Here a node has an effective disadoption rate of $\delta_{i}(1-o_{i})$ and an effective adoption rate of $o_{i}\left( \sum_{\mcal{N}_i^A}\beta_{ij} x_j +\beta_{ii} \right)$. 
\begin{assumption}\label{ass:beta}
It holds that $\beta_{ii}>0, \forall{i}$. 
\end{assumption}
Assumption \ref{ass:beta} ensures that should a subpopulation have a high opinion of the product, some of the subpopulation will adopt the product even if no network neighbors have adopted; i.e. 
that a subpopulation never ignores their opinion. 

The primary opinion dynamic model that will be considered in conjunction with the adoptive spread model in \eqref{eq:prod} is the canonical Abelson model, which in the 1960s laid the foundation for the study of opinion dynamics \cite{abelson1964mathematical} and which has also been studied in the controls community as the consensus protocol \cite{mesbahi2010graph}. { Here the opinion dynamic is treated as occurring at a subpopulation level, which is not commonly done in the opinion dynamics literature. The communities discussed in this case are conceptually similar to the marketing theory of brand communities \cite{zaglia2013brand,muniz2001brand}, where the community shares a common identity which impacts how they respond to a brand or product. As these communities have increased within group communication with regards to the product and the individual agents generate self image through membership, the agents would be expected to have similar opinions within a community. As such we here use the mean opinion to approximate the behavior of the group. Agents in the network can take part in multiple communities with regard to a given product \cite{bagozzi2012customer}, here we assume each agent is modeled as being part of their dominant community.} 

The modified Abelson dynamics follow
 \begin{align}
 	\dot{o}_i &= g_i(x,o) \equiv \sum_{j\in\mcal{N}_i^O} w_{ij}^{o}(o_j-o_i)+w_{i}^{x}\left({\gamma_{i}} x_{i}-o_{i}\right) \label{eq:smod2},
 \end{align}
 where $w_{ij}^o\geq 0$ is the weight on the opinion network, {$w_{i}^{x}\geq0$ {is a weight that represents the responsiveness of the community to their opinion, and $\gamma _{i}\in[-1,1]$ is a scaling factor describing response of a community to adoption, which is influenced by product quality as well as how a product interacts with a community identity, i.e. a very good Android might not affect a community of iPhone users.}} The final term of the modified opinion dynamics, $w_{i}^{x}\left({ \gamma_i}x_{i}-o_{i}\right)$, captures that a node's behavior must impact their opinion. { The relationship between behavior and opinion, especially at the scale of a population, is difficult to capture and is likely highly nonlinear. Understanding the interplay between opinion and behavior is a central question in the social sciences, including such concepts as cognitive dissonance \cite{festinger1962theory} and attitude-behavior consistency \cite{petty2018attitudes, smith1983attitude} as well as many others. For the purpose of this initial exploration of the coupling of opinion and behavior, we draw inspiration from the model of Taylor \cite{taylor1968towards}, and treat the behavior of the population as a source of external information which influences the opinion of the population.} Taken as a whole, \eqref{eq:smod2} models the fact that a node's opinion is affected by its network neighbors' opinions and its own adoption level.
 
 {
 \begin{assumption}\label{ass:w}
 There exists $i$ such that $w_{i}^{x}>0$.
 \end{assumption}
 Assumption \ref{ass:w} ensures that there is a coupling between the adoption state and the opinion somewhere in the network; essentially that there exists a node that cannot have adopted a product without having their opinion affected by the adoption. Some proofs require the stronger claim that $w_{i}^{x}>0,~\forall i$, which is reasonable given the interpretation of each node as a subpopulation defined by their reaction to a product.} 
{
 \begin{assumption}\label{ass:g}
 It holds that $\gamma_{i}>0,~\forall i$.
 \end{assumption}
 Assumption \ref{ass:g} captures that this paper describes the behavior of desirable products. One could imagine a poorly designed product could have a negative impression on opinion, which is related to \cite{ruf2019antag} where antagonism in the social network can cause negative opinions. 
 \begin{assumption}\label{ass:og}
 The opinion graph, $\mathcal{G}_{O}$ is strongly connected. 
 \end{assumption}
 }

Translating the opinion into vector form, shows that the opinion dynamic satisfies 
\vspace{-1ex}
\begin{equation}
\dot{o}=W{\Gamma}x-(\mathcal{L}_{o}+W)o,\label{eq:ovec}
\end{equation}
where $\mathcal{L}_{O}$ is the weighted in-degree graph Laplacian of the opinion network, $W=\mathrm{diag}\left(w_{i}^{x}\right)$, and {$\Gamma=\mathrm{diag}\left(\gamma_{i}\right)$}. The Laplacian $\mathcal{L}_{O}=D-A$ where $D=\mathrm{diag}(d_{i})$ is the in-degree matrix and $A$ is the adjacency matrix of the opinion graph.

By combining \eqref{eq:prod} and \eqref{eq:smod2}, we have that the combined adoption-opinion dynamic follows
\vspace{-1.5ex}
\begin{align}
\dot{x}_{i}&=-\delta_i x_i (1-o_i) + (1-x_i)o_i\left( \sum_{\mcal{N}_i^A}\beta_{ij} x_j +\beta_{ii} \right), \\
\dot{o}_i &= \sum_{j\in\mcal{N}_i^O} w_{ij}^{o}(o_j-o_i)+w_{i}^{x}\left({\gamma_{i}}x_{i}-o_{i}\right) \label{eq:fullmodel}.
\end{align}
It is assumed the initial conditions $x_i(0),o_i(0) \in [0,1] \ \forall i$ are known. As will be shown in the subsequent section, $x_i(0),o_i(0) \in [0,1] \ \forall i$ implies $x_i(t),o_i(t) \in [0,1] \ \forall i,t\geq 0$. Hence, $x_i(t)$ and $o_i(t)$ are functions from $[0,\infty)$ to $[0,1]$. When convenient, we denote the aggregate $2N$-state vector by $z = [x^T,o^T]^T$.  

\vspace{-1.5ex}

\section{Analysis}\label{sec:analysis}

For the coupled adoption opinion model in \eqref{eq:fullmodel}, each $x_{i}$ represents the proportion of the $i$th subpopulation  that has adopted the product and each $o_{i}$ is a scaled average opinion of the $i$th subpopulation. Consequently, the proposed model is only meaningful for $x_{i},o_{i} \in [0,1]$. As such we first establish well-posedness of the model. 

\begin{proposition}
	For the model in \eqref{eq:fullmodel}, if $z(0)\in[0,1]^{2N}$, then $z(t) \in [0,1]^{2N}$ for all $t \geq 0$.
\end{proposition}
\begin{proof}
	Observe that \eqref{eq:fullmodel} is a system of polynomial ODEs over the compact state space $[0,1]^{2N}$. This implies that the system of ODEs in \eqref{eq:fullmodel} is Lipschitz on $[0,1]^{2N}$ and as such the solutions $z_i(t)$ are continuous for all $i\in\{1,\ldots,2N\}$. 
    
    Suppose the proposition is not true. Then there is an index $i\in\{1,\ldots,2N\}$ such that $z_i(t)$ is the first state to go outside $[0,1]$. Consider the case where $i\in\{1,\ldots,N\}$, i.e. the adoption variable $x_{i}$ leaves $[0,1]$. If $x_{i}$ becomes negative then there exists a time $s_0 > 0$ such that $x_i(s_0) = 0$, $\dot{x}_i(s_0) < 0$, $z_j(t) \in [0,1] \ \forall t\in[0,s_0]$ and $\forall j\neq i$. However by \eqref{eq:fullmodel}, 
    \vspace{-1ex}
    \begin{equation}
    	\dot{x}_i(s_0) = o_i\left( \sum_{\mcal{N}_i^A}\beta_{ij} x_j +\beta_{ii} \right) \geq 0,
    \end{equation}
giving a contradiction. To show $x_i$ cannot exceed one, we apply similar arguments and observe that 
\vspace{-1ex}
	\begin{equation}
    	\dot{x}_i(s_0) = -\delta_i(1-o_i) \leq 0.
    \end{equation}
This equality would contradict $\dot{x}_i(s_0) > 0$, which is required for $x_i$ to exceed one.

Analogous arguments apply to show that when $i\in\{N+1,\ldots,2N\}$ $z_{i}$ cannot leave $[0,1]$, i.e. that $o_{i-N}$ cannot go below zero nor above one. { In particular, if $\gamma_{i}\leq 1, \forall i$ then $o_{i-N}$ can not go above one and if $\gamma_{i}\geq 0, \forall i$ then $o_{i-N}$ can not go below zero.}
\end{proof}

Having shown the well-posedness of the adoption model, we now discuss properties of the adoptive spread model by considering the partial derivatives of the function in \eqref{eq:prod}. Note
\vspace{-1ex}
\begin{align}
	\frac{\partial{f_i}}{\partial x_i} &= -\delta_i(1-o_i) - o_i\left(\sum_{\mcal{N}_i^A} \beta_{ij} x_j +\beta_{ii}\right)  \label{eq:dfdxi},
\end{align}
which is always negative when $z\in[0,1]^{2N}$ since $\beta_{ij},\delta_{i}~\geq~0$
{ and $\beta_{ii} >0$ by Assumption \ref{ass:beta}}. 
The other set of partial derivatives with respect to $x$ is  
\begin{align}
	\frac{\partial{f_i}}{\partial x_j} &= \begin{cases} (1-x_i)o_i\beta_{ij} \ &\text{if} \ j\in\mcal{N}_i^A, j\neq i \\ 0 &\text{if} \ j\notin \mcal{N}_i^A \cup \{i\}, \end{cases} \label{eq:dfdxj}
\end{align}
which is always non-negative when $z\in[0,1]^{2N}$ as $\beta_{ij}\geq 0$.
We also have 
\vspace{-1ex}
\begin{align}
	\frac{\partial{f_i}}{\partial o_i} &= \delta_i x_i + (1-x_i)\left(\sum_{\mcal{N}_i^A} \beta_{ij} x_j +\beta_{ii}\right) ,\label{eq:dfdoi}
    \end{align}
    which is always non-negative when $z\in[0,1]^{2N}$ since $\beta_{ij},\delta_{i} \geq 0$. Finally,
    \vspace{-2ex}
    \begin{align}
	\frac{\partial{f_i}}{\partial o_j} &=  0 \ \forall j \neq i.\label{eq:dfdoj}
\end{align}
As in the classic SIS epidemic model, the adoption of network neighbors encourages the consumer to adopt. In the new coupled model, the opinions of the consumers modify the impact of adoption in \eqref{eq:dfdxi} and encourage adoption via \eqref{eq:dfdoi}. 

Consider the behavior of the opinion dynamic model via the partial derivatives of the function in \eqref{eq:smod2}.
\begin{align}
\frac{\partial{g_i}}{\partial x_i} &= w_{i}^{x} { \gamma_{i}}
\label{eq:dg'dxi} \\
\frac{\partial{g_i}}{\partial x_j} &= 0 \ \forall j\neq i\label{eq:dg'dxj} \\
\frac{\partial{g_i}}{\partial o_i} &= -d^{O}_{i}-w_{i}^{x} \label{eq:dg'doi} \\
\frac{\partial{g_i}}{\partial o_j} &= \begin{cases} 1 \ &\text{if} \ j\in\mcal{N}_i^O \ \forall j\neq i\\ 0 &\text{if} \ j\notin \mcal{N}_i^O \cup \{i\}, \end{cases} \label{eq:dg'doj}
\end{align}
where $d^{O}_{i}$ is the (in)degree of the $i$th node in the opinion network. Here the node's adoption state and the opinions of their network neighbors affect the opinion of the node. 

This system { always has an equilibrium at $z^* =0_{2N}$, the case where no one adopts the product and everyone has an opinion equal to zero, which we refer to as the ``flop'' or ``none-adopt'' equilibrium. In the case of a high quality product, when  $\gamma_{i}=1,~\forall i$, then $z^* =1_{2N}$ is also an equilibrium point, i.e. everyone adopts the product and has an opinion equal to one, the ``hit'' or ``all-adopt'' equilibrium. If $\exists i, \gamma_{i}< 1$, then this effects the equilibria as shown in the following lemma. 
\begin{lemma}
If $\exists i, \textrm{ s.t. } \gamma_{i}<1$ and $w_{i}^{x}>0$, then $1_{2N}$ is no longer an equilibrium point. 
\end{lemma}
\begin{proof}
Suppose node $i$ satisfies $ \gamma_{i}<1$ and $w_{i}^{x}>0$. Consider the opinion dynamic of node $i$ at $1_{2N}$.
\begin{align}
\dot{o}_{i}&=\sum_{j\in\mcal{N}_i^O} w_{ij}^{o}(o_j-o_i)+w_{i}^{x}\left(\gamma_{i} x_{i}-o_{i}\right)\\
&=w_{i}^{x}\left(\gamma_{i} x_{i}-o_{i}\right)<0.
\end{align}
\vspace{-2ex}
\end{proof}

Simulations have shown that if $\gamma_{i}<1,~ w_{i}^{x}>0 ~\forall i$ there can exist a stable equilibrium point on $(0,1)^{2N}$. At such equilibria  $ o_{i}^{\ast}$ is close to $\gamma_{i}x_{i}^{\ast}$, where close depends on the system parameters and underlying graph structures. Section \ref{sec:sim} shows a case where $\gamma_{i}x_{i}^{\ast}=o_{i}^{\ast}.$ The parameter $\gamma_{i}$ can also have an impact on the stability of $0_{2N}$ as shown by the following lemma. 
\begin{lemma}\label{lem:loc0}
The equilibrium point $z=0_{2N}$ is locally stable if $\forall i, \delta_{i}>\beta_{ii}$ or if $\forall i, w^{x}_{i}>0, \ \ \delta_{i}>\gamma \beta_{ii},$ where $\gamma=\max_{i} \gamma_{i}.$
\end{lemma}
All stability proofs have been moved to the appendix to facilitate exposition.}

{
We introduce the following notation
\begin{equation}\label{eq:omega}
	\Omega_i(\tau) \equiv \sum_{\mcal{N}_i^A} \beta_{ij}\tau+\beta_{ii} \mbox{ for } i\in\{1,\ldots,N\},
\end{equation}
which captures a node's maximal adoption rate when network neighbors have adoption $x_{j}\leq\tau$. 

\begin{theorem} \label{thm:0global}
If $\delta_{i}>\Omega_{i}(\tau), ~w_{i}^{x}>0, 
\ \forall i$, then $0_{2N}$ is asymptotically stable on $[0,\tau]^{N}\times[0,1]^{N}$ for $\tau<1$ and $[0,1]^{2N}\setminus\{1_{2N}\}$ for $\tau=1$.
\end{theorem} 

 The parameter values that ensure global convergence, $\delta_{i}>\Omega_{i}(1) \ \forall i$, can be interpreted as communities whose rate of disadoption is greater than any potential adoption from network neighbors. 
 For example this can occur in a community whose identity does not allow for the use of a given product, such as a vegan community that will not consume a new meat product, even if all their network neighbors have adopted the product. }

\begin{remark}
The sufficient condition for stability of $0_{2N}$ { on $[0,1]^{2N}\setminus\{1_{2N}\}$} in Theorem \ref{thm:0global} is equivalent to 
\begin{equation*}
\Omega_i{(1)} - \delta_{i} < 0 \ \forall i.
\end{equation*}
which by 
{
the Gershgorin Disc Theorem
}
implies that \begin{equation}\label{eq:s}
\alpha(B-D) < 0, 
\end{equation} 
where $B$ is the matrix of $\beta_{ij}$'s and $D = \text{diag}\left(  \delta_{i} \right)$. 

This is the well-known necessary and sufficient condition for asymptotic stability of the healthy state $0_N$ for the general networked SIS epidemic model \cite{liu2016onthe,bivirus_tac}.  Note that the condition in Theorem~\ref{thm:0global} causes all the Gershgorin discs to be strictly in the left half plane, a sufficient condition for \eqref{eq:s} to hold. Hence, the condition for Theorem \ref{thm:0global} is more stringent than \eqref{eq:s}.

\end{remark}

\renewcommand\arraystretch{1.25}
{\subsection{Behavior of Exceptional Products}
This section considers the behavior of the system when an exceptional product, $\gamma_{i}=1,~\forall i$, is spreading through the network. First, the hit equilibrium $z^{\ast}=1_{2N}$ is characterized and then a class of unstable equilibria is studied.   }
    We consider the behavior of the hit equilibrium $z^{\ast}=1_{2N}$, { which is an equilbrium point if  $\gamma_{i}=1,~\forall i$}. 

   {     \begin{theorem}\label{th:as1}
    If $\Omega_{i}(\tau)>\delta_{i},~{ ~w_{i}^{x}>0, \gamma_{i}=1} \ \forall i$, then $1_{2N}$ is asymptotically stable on $[\tau,1]^{N}\times[0,1]^{N}$ for $\tau>0$ and $[0,1]^{2N}\setminus\{0_{2N}\}$ for $\tau=0$.
    \end{theorem}}
    { When $\tau=0$, the condition in Theorem \ref{th:as1} becomes $\beta_{ii}>\delta_{i}$ which captures that a sufficient level of innovation is required to ensure that a product takes off with no prior adoption.} 

{ Even in the case of an exceptional product, $\gamma_{i}=1,~\forall i$, ~} if the stability conditions presented previously for the global stability of $z^{\ast}=1_{2N}$ or $z^{\ast}=0_{2N}$ are not satisfied, there is the possibility that a third equilibrium point exists for the system. It is also possible that $z^{\ast}=1_{2N}$ or $z^{\ast}=0_{2N}$ are unstable. One class of these equilibria is studied and is shown to be unstable. 
\begin{lemma}\label{lem:eqast}
If { $\gamma_{i}=1~\forall i$ and} there exists a $z^{\ast}$ such that for all $i$ it holds that $\delta_{i}=\sum_{\mcal{N}_i^A} \beta_{ij} x^{\ast}_{j} +\beta_{ii}$, $x^{\ast}_{i}=o^{\ast}_{i}$ and $\sum_{j\in\mathcal{N}_{o}^{i}} x_{j}-x_{i}=0$ then $z^{\ast}$ is an equilibrium point.
\end{lemma}

\begin{theorem}\label{th:unstab}
If the equilibrium described in Lemma \ref{lem:eqast} exists and $w_{i}^{x}>0,~\forall i$,
it is unstable. 
\end{theorem}

Since the hit and flop equilibria ($1_{2N}$, $0_{2N}$, respectively), { can} satisfy the conditions of Lemma \ref{lem:eqast}, { together with the local stability result for $1_{2N}$ presented in the Appendix} we obtain the following.
\begin{corollary}
Suppose the opinion graph is strongly connected and { $\gamma_{i}= 1,~w_{i}^{x}>0, \forall i$}. If $\delta_{i}=\beta_{ii},~\forall i$, then $0_{2N}$ is unstable; if $\delta_{i}>\beta_{ii},~\forall i$, then
$0_{2N}$ is locally stable; and if $\delta_{i}>\Omega_i{(1)},~\forall i$, then
$0_{2N}$ is asymptotically stable.
\end{corollary}
\begin{corollary}
Suppose the opinion graph is strongly connected and { $\gamma_{i}= 1,~w_{i}^{x}>0, \forall i$}. If $\delta_{i}=\Omega_i{(1)},~\forall i$, then $1_{2N}$ is unstable; if $\Omega_i{(1)}>\delta_{i},~\forall i$, then
$1_{2N}$ is locally stable; and if $\beta_{ii}>\delta_{i},~\forall i$, then
$1_{2N}$ is asymptotically stable.
\end{corollary}

It is possible that the equilibrium in Lemma \ref{lem:eqast} exists in $(0,1)^{2N}$. If such an equilibrium exists and the opinion graph is strongly connected then this equilibrium is unstable, causing the system behavior to be dependent on the initial condition as is explored further in the Simulation Section. A summary of the stability results of this section is given in Table \ref{table:stable_table}.

\begin{table}
\centering
\resizebox{\columnwidth}{!}{
\begin{tabular}{c|c|c|c}
$z^{\star}$
&Unstable&Local Stable & Asymptotic Stable \\ \hline  
$0_{2N}$&$\delta_{i}=\beta_{ii}$&$\delta_{i}>\beta_{ii}$&$ \delta_{i}>\Omega_i{(1)} $\\
$1_{2N}$&$\Omega_i{(1)} = \delta_{i}$&$\Omega_i{(1)} > \delta_{i}$&$\beta_{ii}>\delta_{i}$\\
\end{tabular}}
\caption{Summary of stability conditions: recall from \eqref{eq:omega} $\Omega_i{(\tau)} \equiv \sum_{\mcal{N}_i^A} \beta_{ij}{\tau} +\beta_{ii}$.
}
\label{table:stable_table}
\end{table}

\section{Bounded Confidence Model}
In this section, we extend our consideration to the bounded confidence opinion dynamic model. The bounded confidence model is an extension of the Abelson opinion dynamic model \cite{canuto2008eulerian,blondel2010continuous},
which when coupled with the adoption dynamic is as follows:
\begin{equation}
\dot{o}_{i}=g_{i}(x,o)=\sum_{j\in\mcal{N}_i^O} p(o_{j},o_{i})(o_j-o_i)+w_i^x({\gamma_{i}}x_{i}-o_{i}) \label{eq:ob},
\end{equation}
\noindent where
\begin{equation}
p(o_{j},o_{i}) = \begin{cases}
            w_{ij}^o &\mbox{if} \  \|o_{j}-o_{i}\|<\xi \\
            0 &\mbox{if} \ \text{else}.	
        \end{cases}
\end{equation}
Under the bounded confidence model, nodes will sever a link in the opinion graph if the nodes have sufficiently different opinions and maintain or reintroduce the link if the respective opinions are closer than $\xi$. This behavior is essentially a state dependent switch between opinion graph topologies. These opinion graphs may not be connected, to the point where each node may have no neighbors in the opinion graph. However the structure of the coupling with the adoption dynamic ensures that the conditions for the asymptotic equilibria $z^* \in\{0_{2N},1_{2N}\}$ are the same. 

As { is detailed in the Appendix}, we consider arbitrary switching behavior in the bounded confidence model and as such the domain over which stability is proven must be modified from $[0,1]^{2N}\setminus \{1_{2N}\}$ to $[0,1)^{2N}$. It is likely that the specific interaction pattern of the bounded confidence model will allow stronger results to be shown under additional constraints, however such a consideration is left for future work.

\begin{theorem}\label{th:as0bc}
If $\delta_{i}>\Omega_i{(1)} \ \forall i$, the opinion graph is undirected, and $\Omega_i{(1)}=w_{i}^{x}, \ \forall i$ then $0_{2N}$ is uniformly asymptotically stable on $[0,1)^{2N}$ under the bounded confidence opinion dynamic.
\end{theorem}

\begin{theorem}\label{th:as1bc}
    If $\beta_{ii}>\delta_{i},~ { \gamma_{i}=1} \ \forall i$, the opinion graph is undirected, and if $\delta_{i}=w_{i}^{x}, \ \forall i$, then $1_{2N}$ is asymptotically stable on $(0,1]^{2N}$ under the bounded confidence opinion dynamic.
    \end{theorem}
    { Simulation have shown that} the system exhibits asymptotic stability without the conditions which enforce symmetry for the matrices $P_{i}$, i.e. if 
    the opinion graph is 
    directed and 
    $\Omega_{i}{(1)\neq} w_{i}^{x}, \ \forall i,$ the system still seems to be asymptotically stable. The theoretical characterization of such a case will be left for future work.

\section{Simulation}\label{sec:sim} 

In this section, we consider the consensus-adoption model in \eqref{eq:fullmodel} and show instances through numerical simulations where it exhibits behaviors that have been observed in empirical studies of adoptive spread \cite{Centola1194,centola2018experimental}. In particular, we focus on behaviors shown in the study of complex contagions, which aims to understand the effect of social reinforcement on adoption dynamics.  We show that the model exhibits different complex contagion behavior under certain opinion topologies by highlighting the importance of weak ties. We also show parameter regimes where our model behavior is dependent on initial conditions,  showing threshold-like behavior similar to that of a tipping point. We conclude the section by observing the existence of stable equilibria not described in the theoretical results.

\subsection{Complex Contagion: The Impact of Opinion}
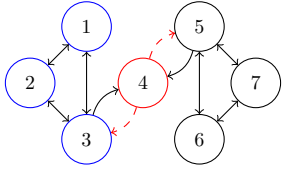
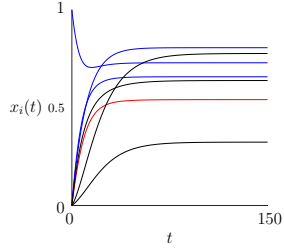
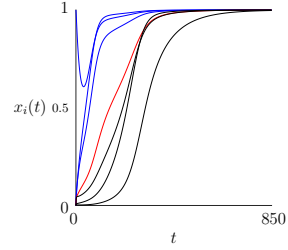
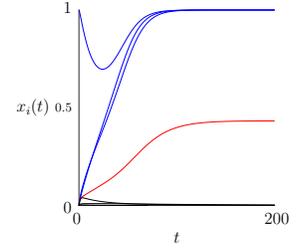
\begin{figure*}[h]
\centering
\begin{subfigure}[b]{.65\columnwidth}
\centering
\begin{tikzpicture}
		[
		->,
		auto,
		node distance=0.5cm,
		every text node part/.style={align=center},
		scale=0.25, every node/.style={scale=.75}
		]
		\tikzstyle{every state}=[fill=none,text=black]
		\node[
		state,draw=blue
		] (n0) at(-3,3)
		{$1$};
		\node[
		state,draw=blue
		] (n1) at(-6,0)
		{$2$};
		\node[
		state,draw=blue
		] (n2) at(-3,-3)
		{$3$};
		\node[
		state,draw=red
		] (n3) at(0,0)
		{$4$};
		\node[
		state,draw=black
		] (n4) at(3,3)
		{$5$};
        		\node[
		state,draw=black
		] (n5) at(3,-3)
		{$6$};
        		\node[
		state,draw=black
		] (n6) at(6,0)
		{$7$};
		\path (n1) edge[<->] (n0)
        (n0) edge[<->] (n2)
        (n1) edge[<->] (n2)
        (n2) edge[->,bend left] (n3)
        (n3) edge[->,bend left,red,dashed] (n2)
        (n4) edge[->,bend left] (n3)
        (n3) edge[->,bend left,red,dashed] (n4)
        (n4) edge[<->] (n5)
        (n4) edge[<->] (n6)
        (n5) edge[<->] (n6);
		\end{tikzpicture}
		\vspace{2.5ex}
        \caption{Underlying Graph Structure $\mathcal{G}_{P}$: \\ 
        $\mathcal{G}_{O} \equiv \mathcal{G}_{P}$ (for Fig. \ref{fig:cc_same}) and has the \\
        red dashed edges removed for Fig. \ref{fig:cc_del}}
      \label{fig:cc_graph}
        \end{subfigure} \hspace{-.5cm}
\begin{subfigure}[b]{.45\columnwidth}
\centering
\includegraphics[width=1\columnwidth]{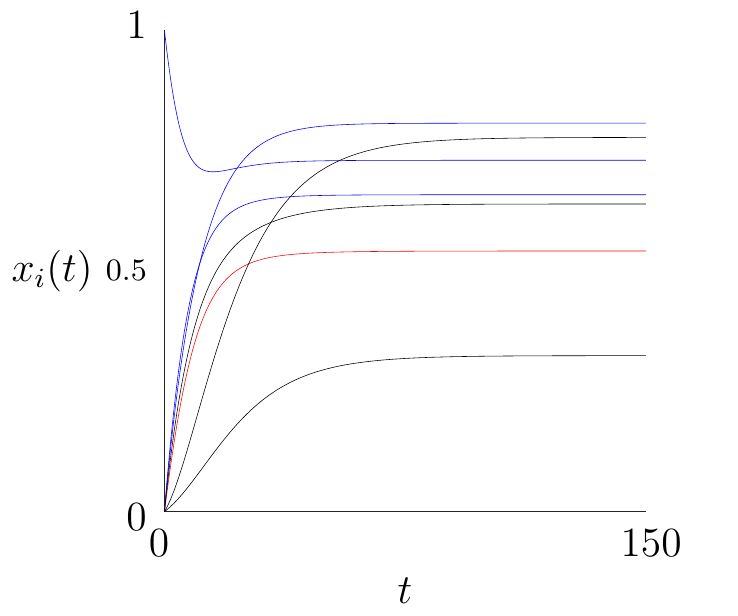}
\caption{Adoption when there is no coupling with opinion}
      \label{fig:cc_noc}
\end{subfigure} 
\begin{subfigure}[b]{.45\columnwidth}
\centering
\includegraphics[width=1\columnwidth]{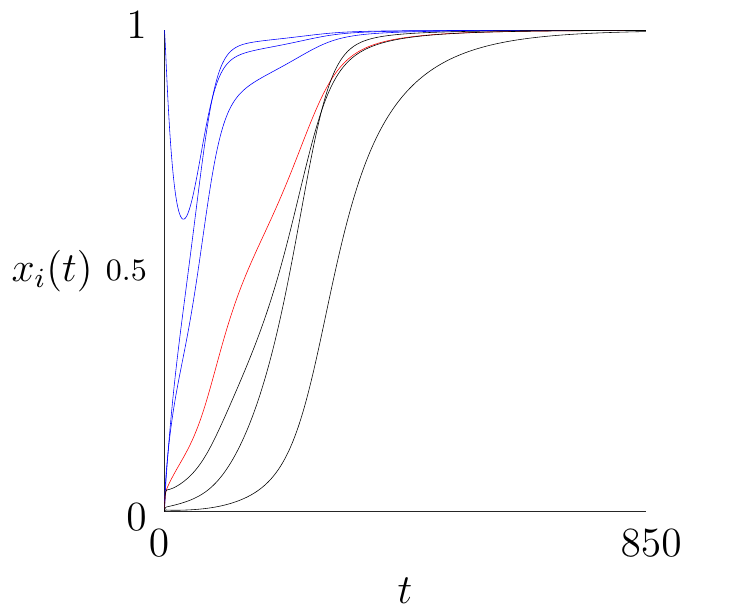}
\caption{Adoption when both graphs are identical}
      \label{fig:cc_same}
\end{subfigure}
\begin{subfigure}[b]{.45\columnwidth}
\centering
\includegraphics[width=1\columnwidth]{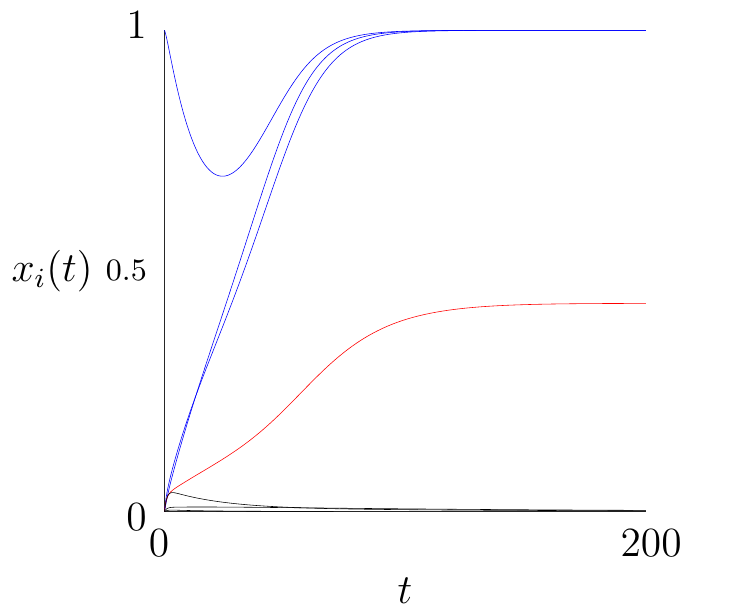}
\caption{Adoption when the red edges are deleted}
      \label{fig:cc_del}
\end{subfigure}
\caption{The adoption of a barbell graph with random parameters run starting from $x_{1}=1$ and $x_{j}=0 \ \forall j\neq 1$ for three separate conditions on the opinion graph. Figure \ref{fig:cc_noc} shows when there is no coupling with opinion (i.e. it is the standard SIS epidemic model). Figure \ref{fig:cc_same} shows when the opinion graph is an unweighted version of the infection graph. Figure \ref{fig:cc_del} shows when the opinion graph has the two red, dashed edges deleted from the adoption graph, making node $4$ an information bottleneck.} \label{fig:ops}
\end{figure*}
In the sociology literature, and specifically the study of complex contagions, a distinction is made between the impact of strong and weak friendship ties in a social network \cite{Centola1194,granweak}. A weak tie is characterized by the lack of many common friendships and low emotional intensity, representing for example a tie with an acquaintance. A strong tie is the opposite, having high intimacy and strong emotional intensity, for example a tie with a family member or close friend. Granovetter found that for information diffusion, such as the availability of jobs, weak ties are vital for the spread of information by serving as bridges between different communities\cite{granweak}. 

While weak ties are important for the spread of simple contagions, like information or a virus; in the case of a complex contagion these weak ties can serve as bottlenecks. This is due to the fact that complex contagions require social reinforcement to spread, which is often missing in the case of a weak tie\cite{Centola1194}. The coupled adoption opinion model presented in \eqref{eq:fullmodel} allows the study of the effect of the underlying graph structure of the opinion dynamic, which captures social reinforcement. Figure~\ref{fig:ops} shows that varying the underlying graph of the opinion dynamic can change a weak tie from a conduit to a bottleneck.

The coupled adoption opinion model is simulated on a $7$-node network with the adoption graph $\mathcal{G}_{P}$ shown in Figure~\ref{fig:cc_graph}. The opinion graph $\mathcal{G}_O$ is varied between Figures \ref{fig:cc_noc}-\ref{fig:cc_del}, however across all opinion graphs the edge weights follow $w_{ij}^{o}=1,~\forall i,j\in\{1,\dots,n\}$. The randomly chosen adoption parameters are as follows: 
\vspace{-3ex}

\footnotesize
$$B=\begin{bmatrix}0.0665&0.0668& 0.0630&0&0&0&0\\
0.0718&    0.0033&    0.0477&         0&         0&         0&         0\\
    0.0281&    0.0521&    0.0549&    0.0641 &        0 &        0 &        0\\
         0 &        0 &   0.0114&    0.0525 &   0.0480 &        0 &        0\\
         0&         0&         0&    0.0250&    0.0646&    0.0432&    0.0575\\
         0&         0&         0&         0&    0.0112&    0.0050 &   0.0346\\
         0&         0 &        0  &       0  &  0.0470&    0.0421&    0.0108\end{bmatrix}$$,
         \normalsize
         
         \noindent $D=\mathrm{diag}(0.0599,0.0208,0.0790,0.0767,0.0773,0.0813, \\ 0.0156)$, and { $\gamma_{i}= 1, \forall i$}. It holds that $\beta_{ii}>\delta_{i}, ~ \forall i\neq 1 $, suggesting that the population is likely to adopt the innovation. This system was initialized at $x(0)=[1 ,0 ,0 ,0 ,0 ,0 ,0]^{T}$ and $o(0)=[0.8279, 0.2410, 0.7215, 0.9841, 0.6457, 0.5573, 0.9630]^{T}.$

    As can be seen from Figure~\ref{fig:ops}, varying information topologies causes very different behavior in the model. If there is no coupling with an opinion, shown in Figure~\ref{fig:cc_noc}, then the innovation spreads from node $1$ to the entire population with varying success. If the adoption graph and the opinion graph are identical, as in Figure~\ref{fig:cc_same}, then the innovation is able to completely spread throughout the graph. 
    
    Figure~\ref{fig:cc_del} shows the impact of an information bottleneck: here node $4$ receives information about the opinion of node $3$ and node $5$ but does not spread information as the outgoing links have been deleted. This stops information about the innovation from nodes $1$ to $3$ from spreading to nodes $5$ to $7$. Consequently, the information bottleneck, node $4$, is also an adoption bottleneck; i.e. the information topology prevents the product from spreading into the right half of the graph.   
    
We have shown the impact that opinion can have on adoption behavior in the model presented in \eqref{eq:fullmodel}, suggesting that this model is a valuable tool for understanding the social reinforcement effects typically studied via threshold models.  
    \vspace{-.2cm}

\subsection{Tipping Point-Like Behavior}
A central concept in the spread of innovations 
is the tipping point \cite{centola2018experimental,schelling1971dynamic,granovetter1978threshold}. A tipping point on a population level is a fraction of adopters which determines the prevalence of the product. If the adoption level is under the tipping point, the product does not spread to the whole population while if the adoption is above the tipping point the product spreads to the whole population. Lemma~\ref{lem:eqast} describes an equilibrium point which occurs if the coupled adoption opinion model satisfies for all $i$ that $\delta_{i}=\sum_{\mcal{N}_i^A} \beta_{ij} x^{\ast}_{j} +\beta_{ii}$, $x^{\ast}_{i}=o^{\ast}_{i}$ and $\sum_{j\in\mathcal{N}_{o}^{i}} x_{j}-x_{i}=0$. Theorem \ref{th:unstab} shows that this equilibrium is unstable if the opinion graph is strongly connected and { $\gamma_{i}= 1, \forall i$}. Simulations show that if such an unstable equilibrium exists in $(0,1)^{2N}$, it will induce initial condition dependent behavior in the model, similar to a tipping point. 

This dependence is shown in Figure~\ref{fig:mc}. The adoption graph $\mathcal{G}_{P}$ is a complete graph on four nodes and the opinion topology $\mathcal{G}_{O}$ is taken to be a complete graph with edge weights $w_{ij}^{o}=1,~\forall i,j$. The system parameters are $$B=\begin{bmatrix}{0.1}&{0.25}&{0.3}&0.35\\
    0.15&{0.05}&0.3&{0.3}\\
    {0.5}&0.3&{0.1}&0.3\\
    0.2 &{0.1}&{0.1}&{0.2}\end{bmatrix},$$ $D=\mathrm{diag}(.5,.4,.6,.3)$, { and $\gamma_{i}= 1, \forall i$.} With these parameters, $x_{i}=o_{i}=.5,~\forall i$ is an equilibrium as described in Lemma~\ref{lem:eqast}. The system was run $10,000$ times with randomly selected initial conditions {under both the consensus and the bounded confidence model with $\xi=.01$}. The results are shown in Figure~\ref{fig:mc}. When the initial condition is sufficiently high, the system converges to the hit equilibrium, $1_{2N}$. Conversely if the system has an initial condition that was sufficiently low, the system converges to flop equilibrium, $0_{2N}$. {In the bounded confidence case, one or more nodes can break away from their neighbors, resulting in a split equilibrium.} Under the current system parameters, Figure~\ref{fig:mc} suggests that initial adoption has a larger impact on the resultant equilibrium. { This matches what one might expect based on the conditions of Theorem \ref{thm:0global} and \ref{th:as1},} however further work is required to completely characterize the tipping point-like behavior of the system {and how a system which shows these behaviors would react to control effort}. 

\begin{figure}
	\centering
		\includegraphics[width=.48\columnwidth]{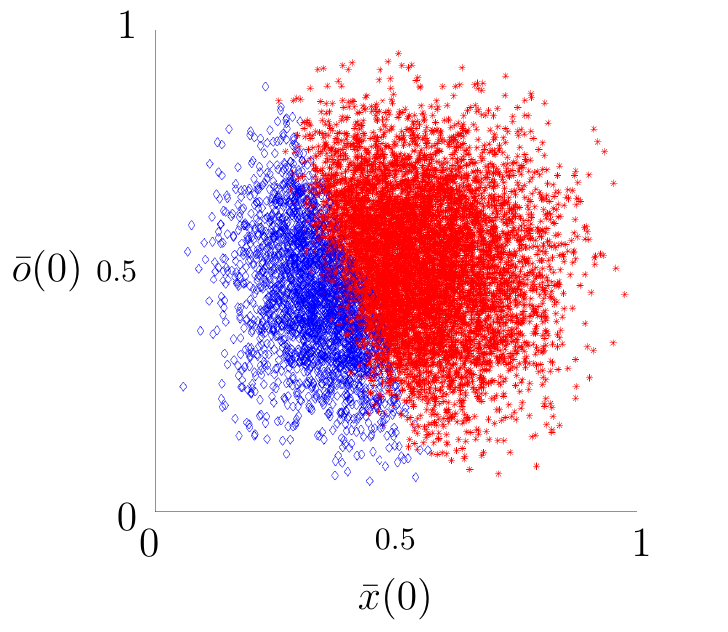}
		\includegraphics[width=.48\columnwidth]{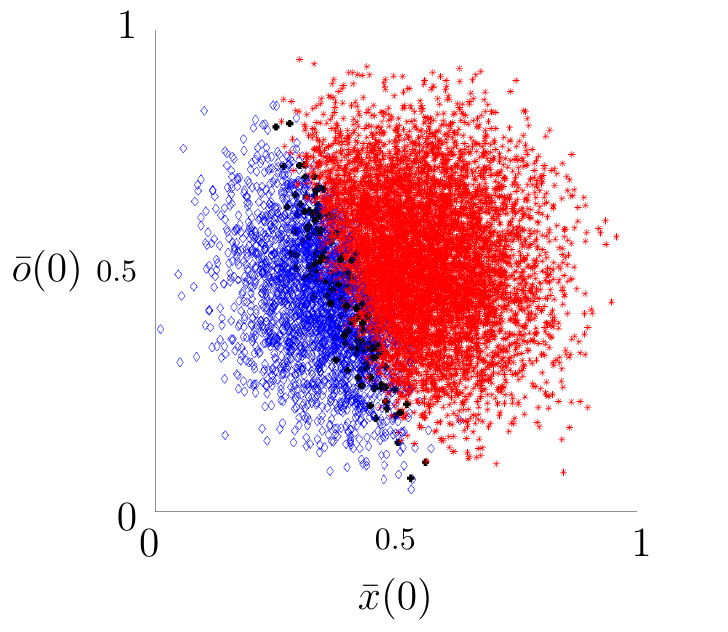}
        \caption{Simulation of the effect of initial conditions on the system equilibria in the presence of an unstable equilibrium. The x-axis shows mean initial adoption, the y-axis mean initial opinion. The blue diamonds and red stars represent when the system converged to $0_{2N}$ and $1_{2N}$, respectively. The right figure shows what happens when the bounded confidence model is run as the opinion dynamic, resulting in systems that do not converge to $0_{2N}$ or $1_{2N}$, shown by black crosses. 
        } \label{fig:mc}
 
\end{figure}

\subsection{Stable Equilibrium}
While this paper has focused on the analytical characterization of the hit, the equilibrium at $1_{2N}$, and the flop, the equilibrium at $0_{2N}$, it is also possible that there exists a stable equilibrium in the open interval $(0,1)^{2N}$. While the characterization of such equilibria will require future work, it is possible to discuss their  existence through simulation. 

One way that these stable equilibria can exist is if there is a mix between nodes which satisfy $\delta_{i}>
\Omega_i{(1)}$ 
and nodes which satisfy $\delta_{i}<\beta_{ii}$. 

Figure \ref{fig:stabstar} considers a $5$-node star graph topology with node $1$ as the center node. When $D=\mathrm{diag}(5,0.1,0.1,0.1,0.1)$ and $$B=\begin{bmatrix} 0.1& 0.2& 0.2& 0.2& 0.2\\
0.01& 0.15& 0& 0& 0\\
0.01& 0& 0.15& 0& 0\\
0.01& 0& 0& 0.15& 0\\
0.01& 0& 0& 0& 0.15
\end{bmatrix}$$ 
the system follows the trajectory in Figure \ref{fig:stabstar}. The center node satisfies $\delta_{1}>
\Omega_1$
while the peripheral nodes satisfy $\delta_{k}<\beta_{kk},\forall k\neq 1$. Here $\delta_{1}$ is large enough that the nodes will not converge to $1_{2N}$ but instead reach an equilibrium value in the interval $(0,1)^{2N}$. Figure \ref{fig:stabstar} shows the evolution of the system under a sample initial condition. The equilibrium at 
\begin{align*}
    x=\begin{bmatrix}0.1114 & 0.6829 7   0.6829 & 0.6829 &    0.6829\end{bmatrix}^T, \\  o=\begin{bmatrix}0.4924    & 0.5877 &   0.5877 &   0.5877   & 0.5877\end{bmatrix}^T
\end{align*} 
was shown to be unique by running the system $1000$ times with varying initial conditions. Simulations show that under these parameters the hit and flop equilibria are  both unstable. 
    
\begin{figure}[h]
\includegraphics[width=.48\columnwidth]{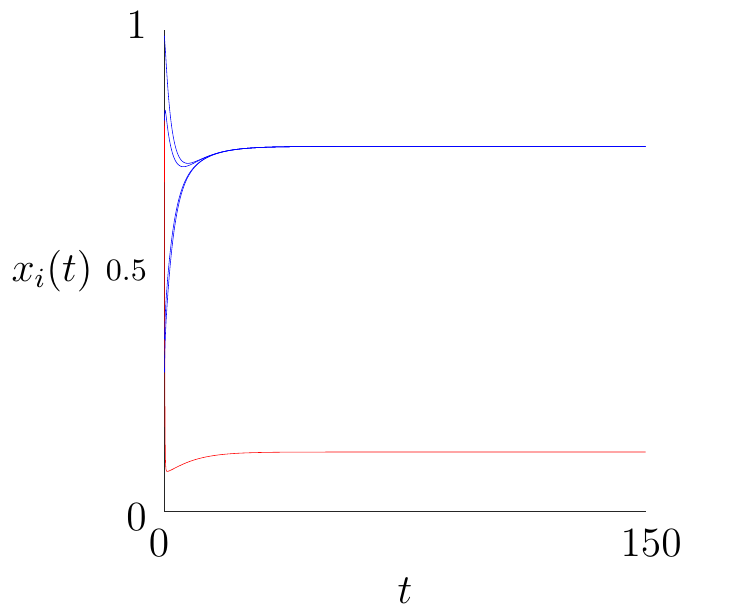}
\includegraphics[width=.48\columnwidth]{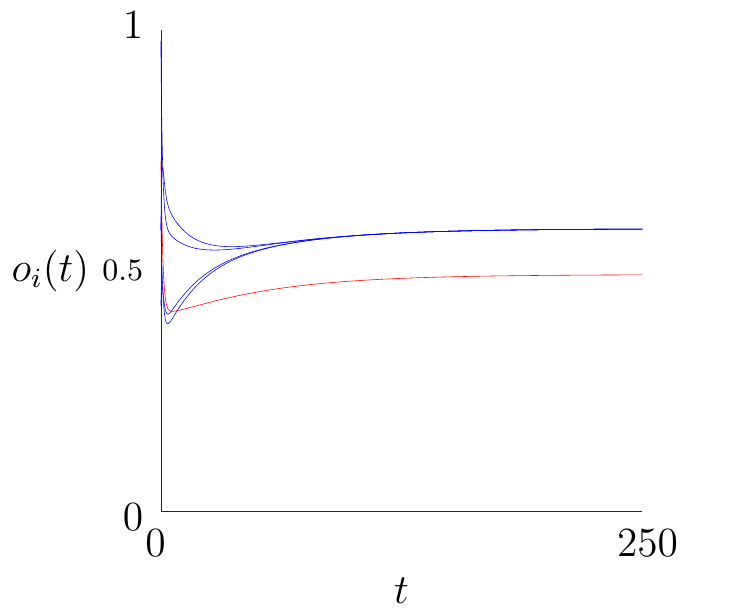}
\caption{System evolution, with adoption on the left and opinion on the right, on the star graph where there is a stable equilibrium in $(0,1)^{2N}$. The red line represents the center node, while the blue lines represent the peripheral nodes.} \label{fig:stabstar}
\end{figure}

{Another way this can happen is if $\gamma_{i}<1,~\forall i$ and the $0_{2N}$ equilibrium is unstable. Figure~\ref{fig:stabcomp} shows the results for a $20$ node complete graph where $B=.0051_{N\times N}+.145I$ (i.e. $\beta_{ii}=.15$), $D=\mathrm{diag}(.1)$, the opinion graph is complete with unit weights, and $\gamma_{i}=.75, ~\forall i$ which was initialized at $.5_{2N}$. At equilibrium $x_{i}^{\ast}=0.718,~\forall i$ and $o_{i}^{\ast}=0.5385=\gamma_{i}x_{i}^{\ast},~\forall i$. }

\begin{figure}[h]
\includegraphics[width=.48\columnwidth]{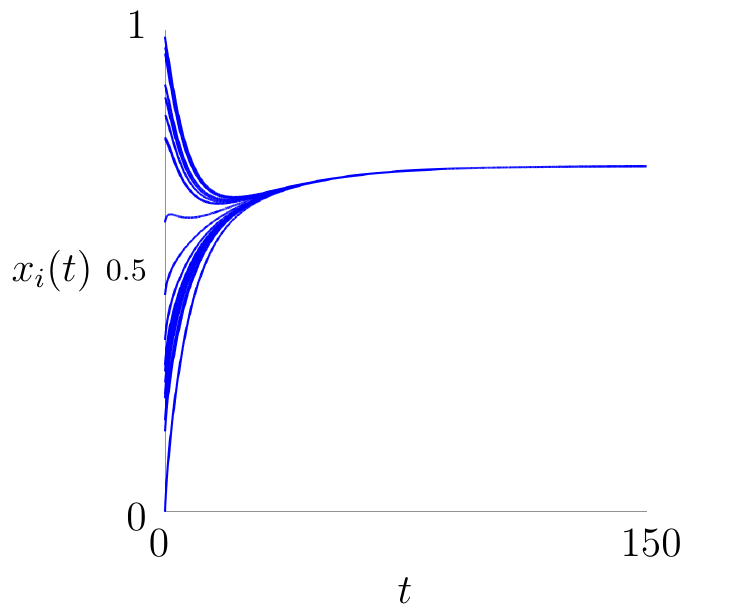}
\includegraphics[width=.48\columnwidth]{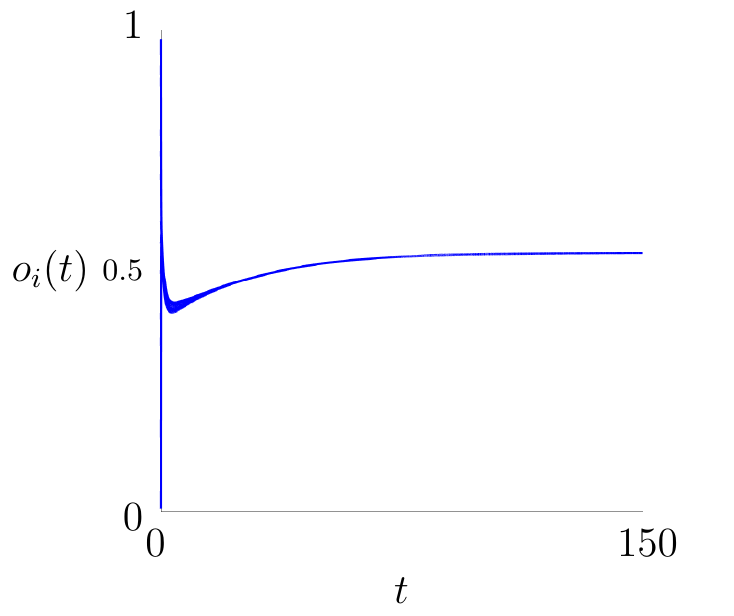}
\caption{System evolution, with adoption on the left and opinion on the right, on a complete graph where there is a stable equilibrium in $(0,1)^{2N}$ caused by $\gamma_{i}<1,~\forall i$.} \label{fig:stabcomp}
\end{figure}

\section{Conclusion}\label{sec:con}
In this paper, we consider a coupled adoption opinion model in which the spread of an epidemic product is influenced by the evolution of an opinion dynamic. The stability of the hit and flop equilibria are shown based on the adoption parameters and the opinion graph. These results are also extended to the case of the bounded confidence opinion dynamic. Finally, it was shown in simulation that the presented model exhibits many of the characteristic behaviors of product spread observed in the sociology literature. 

This paper provides a dynamical systems viewpoint for the coupling between adoption and opinion, and ultimately provides an avenue to deepen the understanding of complex contagion adoption dynamics. Behavior observed from the sociology literature, such as tipping points, were exhibited by the model, pointing to its ability to capture a wide range of real world behaviors. Future work is required to explore and connect these behaviors to real contagion phenomenon { and the possibility of applying control.}

\vspace{-1ex}

\section{Appendix}
In this section, the proofs of the various technical results in the paper are presented. First the stability of the $1_{2N}$ and $0_{2N}$ equilibria are discussed for the original model presented in \eqref{eq:fullmodel}, which requires some concepts from matrix analysis which are discussed below. Unless otherwise stated the discussion follows \cite{horn1990matrix}.

\begin{definition}A matrix $A$ is {\em diagonally dominant} if 
\begin{equation}
|A_{ii}|\geq \sum_{j\neq i} |A_{ij}| \ \forall i.
\end{equation}
{
The matrix is {\em strictly diagonally dominant} if the inequality is strict.
}
\end{definition}

\noindent Consider a diagonally dominant matrix $A$ and let \begin{equation}
J\\=\left\{i\in\{1,2,\dots,n\}:|a_{ii}| > \sum_{j\neq i} |a_{ij}| \right\}.
\end{equation}
Any row $j$ such that $j\in J$ is said to be a strictly diagonally dominant row. 

A diagonally dominant matrix with negative diagonal entries has eigenvalues with non-positive real part by 
{
the Gershgorin Disc Theorem
}
a strictly diagonally dominant matrix with negative diagonal entries has eigenvalues with negative real part by 
{
the Gershgorin Disc Theorem. 
}

\begin{definition}
A matrix A is weakly chained diagonally dominant (WCDD) if it is 
\begin{itemize}
\item diagonally dominant, and 
\item for all $i\notin J$ there is a sequence of nonzero elements of A of the form $a_{ii_{1}},a_{i_{1}i_{2}},\dots,a_{i_{r}j}$ with $j\in J$.
\end{itemize}
\end{definition}
\noindent The second condition can be equivalently expressed as the existence of a walk from $i$ to $j$ on the directed graph of $A$. 
WCDD matrices have the following characterization \cite{shivakumar1974sufficient}:
\begin{lemma}\label{lem:wcdd}
A WCDD matrix is nonsingular.
\end{lemma}
Recall the following condition for Metzler matrices from~\cite{berman1994nonnegative}:
\begin{lemma}
Let $A$ be an irreducible Metzler Matrix. In the following, for two vectors $x,y\in\mathbb{R}^{N}$, $x>y$ 
means $x_i > y_i \forall i$. 

\begin{itemize}\label{lem:metzeig}
\item If there exists $x>0_{N}$ such that $Ax>\lambda x$ for some $\lambda \in \mathbb{R}$, then $\alpha(A)>\lambda.$
\item If there exists $x>0_{N}$ such that $\mu x>Ax$ for some $\mu\in\mathbb{R}$, then $\mu >\alpha(A).$
\end{itemize}

\end{lemma}
\begin{lemma}\label{lem:posdiag}({ Proposition 1 and }Lemma A.1 in \cite{khanafer2016stability})  Suppose that $M$ is an irreducible Metzler matrix such that { $\alpha(M) < 0$ } ($\alpha(M) = 0$). Then,
there exists a positive diagonal matrix $Q$ such that $M^{T}Q +QM$ is negative { definite } (semi-definite).
\end{lemma}
With the mathematical preliminaries concluded, stability of the equilibrium point at $z^\ast=0_{2N}$ can be shown. 

{
\begin{proof}[Proof of Lemma \ref{lem:loc0}]
The Jacobian of the dynamics can be written in block form as:
\vspace{-1.5ex}
\renewcommand\arraystretch{2}
\begin{equation}
J(z) = \left[
\begin{array}{c|c}
\displaystyle\frac{\partial f}{\partial x} & \displaystyle\frac{\partial f}{\partial o} \\
[\smallskipamount]
\hline 
\displaystyle\frac{\partial g}{\partial x} & \displaystyle\frac{\partial g}{\partial o}
\end{array}
\right],
\end{equation}
where first $N$ rows of the Jacobian are given by \eqref{eq:dfdxi}-\eqref{eq:dfdoj} and the second $N$ rows are given by \eqref{eq:dg'dxi}-\eqref{eq:dg'doj}. 
\renewcommand\arraystretch{1.25}

Consider the Jacobian matrix at the equilibrium point \\ 
\vspace{-1ex}
$z^\ast=0_{2N}$: 
\vspace{-2ex}
\begin{equation}
J(z^*) = \left[
\begin{array}{c|c}
\displaystyle \text{diag}\left(-\delta_{i}  \right) & \displaystyle \text{diag}\left(  \beta_{ii} \right)  \\
\hline 
W{ \Gamma} & -(\mathcal{L}_{O}+W) 
\end{array}
\right],\label{eq:sjab}
\end{equation}
using the notation from \eqref{eq:ovec}. To show local stability, this Jacobian must be shown to be Hurwitz.

In the case that $\delta_{i}>\beta_{ii} \ \forall i$, together with the fact that { $\gamma_{i}\leq 1,~ \forall i$ } and the graph Laplacian is diagonally dominant, $
J(z^*)$ is diagonally dominant. The first $N$ rows of $
J(z^*)$ are strictly diagonally dominant  while the diagonal dominance of the second $N$ rows of $
J(z^*)$ depends on $\gamma_{i}$ and $w_{i}^{x}$. If $\gamma_{i}=1$ the row is diagonally dominant, if $\gamma_{i}<1$ the row is strictly diagonally dominant. If $w_{i}^{x}=0$ the row is diagonally dominant; otherwise, it is strictly diagonally dominant. However, by Assumption \ref{ass:w} $\exists j \textrm{ s.t. } w_{j}^{x}>0$ and by Assumption \ref{ass:og} the opinion graph is strongly connected. Therefore there is a path from row $j$ to all rows $k\in\{N+1,N+2,\dots,2N\}$.  

So the Jacobian is WCDD and therefore nonsingular by Lemma \ref{lem:wcdd}.

Since the diagonal elements of $
J(z^*)$ are negative, the above argument combined with 
{
the Gershgorin Disc Theorem
}
shows that the Jacobian is Hurwitz.

If $\delta_{i}>\gamma\beta_{ii} \ \forall i$, there exists a $\rho>1$ such that $\delta_{i}>(\rho\gamma)\beta_{ii} \ \forall i$. As the opinion graph is strongly connected the Jacobian is irreducible allowing the use of Lemma \ref{lem:metzeig}. Consider a vector $y$ of the form $y=\begin{bmatrix} {1}_{N}^{T} &\rho\gamma 1_{N}^{T}
		\end{bmatrix}^{T}.$
		Then consider the matrix product $J(z^{\ast})y$. The first $N$ entries of $J(z^{\ast})y$ follow:
		\begin{equation}
		-\delta_{i}+\beta_{ii}\rho\gamma <0.\\
		\end{equation}
		The last $N$ entries of $J(z^{\ast})y$ follow
		\begin{equation}
		 \gamma_{i} w_{i}^{x}-d_{i}^{O}\rho\gamma-w_{i}^{x}\rho\gamma + \sum_{\mcal{N}_i^O} \rho\gamma 
		=(\gamma_{i}-\rho\gamma) w_{i}^{x}<0.
		\end{equation}
		Therefore by Lemma \ref{lem:metzeig} the Jacobian is Hurwitz.
\end{proof}
}
Having shown local stability of the flop equilibrium $z^{\ast}=0_{2N}$, we move to showing asymptotic stability. 

Before proving Theorem \ref{thm:0global}, we introduce some lemmas that will be required for the proof. Consider the matrix \begin{equation}P=\begin{bmatrix}
-\bar{B} & \bar{B} \\
W{\Gamma} & -(\mathcal{L}_{o}+W)
\end{bmatrix}\label{eq:Pmat}\end{equation}
where $\bar{B}=\mathrm{diag}\left(\Omega_i{(\tau)}\right)$. 
\begin{lemma}\label{lem:Pmat}
If $\delta_{i}>\Omega_i{(\tau)}, ~\forall i$ { and the system is in $[0,\tau]^{N}\times[0,1]^{N}$}, then the coupled dynamic in \eqref{eq:fullmodel} satisfies $\dot{z}\leq Pz$.
\end{lemma}
\begin{proof}
Consider the adoption dynamic for the case where the adoption parameters satisfy $\delta_{i}>\Omega_i{(\tau)}, \forall i$ and when the state satisfies $x_i\neq 0$ and $o_i\neq 1$: %
\vspace{-2ex}
\begin{align}
\dot{x}_{i}&= -\delta_i x_i (1-o_i) + (1-x_i)o_i\left( \sum_{\mcal{N}_i^A}\beta_{ij} x_j +\beta_{ii} \right)\\
	&\leq -\delta_i x_i (1-o_i) + (1-x_i)o_i\Omega_i{(\tau)}\\ 
    	&< -\Omega_i{(\tau)} x_i (1-o_i) + (1-x_i)o_i\Omega_i{(\tau)}\\
        &=(o_{i}-x_i)\Omega_i{(\tau)}.
\end{align}
Consider now when $x_i=0$ or $o_i=1$. In the case that $o_{i}=1$ 
then $\dot{x}_{i}\leq(1-x_{i})\Omega_i{(\tau)}$ and $(o_{i}-x_i)\Omega_i{(\tau)}=(1-x_{i})\Omega_i{(\tau)}$. In the case that $x_{i}=0$ 
then $\dot{x}_{i}\leq o_{i}\Omega_i{(\tau)}$ and $(o_{i}-x_i)\Omega_i{(\tau)}=o_{i}\Omega_i{(\tau)}$. Together it holds that when $x_i= 0$ or $o_i= 1$ 
\vspace{-1ex}
\begin{equation}
\dot{x}_{i}\leq(o_{i}-x_i)\Omega_i{(\tau)}.
\end{equation}

Translating this to matrix form gives that $\dot{z}\leq Pz.$

\end{proof}

\begin{lemma}\label{lem:leqz}
The eigenvalues of $P$ all have non-positive real part. { If there $\exists i,~ \textrm{s.t.}~ \gamma_{i}<1, w_{i}^{x}>0$ then $\alpha(P^{\ast})<0$. Otherwise,} $\alpha(P)=0$. 
\end{lemma}
\begin{proof}
$P$ is diagonally dominant and has negative diagonal entries. Therefore by 

{
the Gershgorin Disc Theorem
}
the real parts of all the eigenvalues are non-positive. { Suppose $\gamma_{k}<1, w_{k}^{x}>0$, then strong connectivity of the opinion graph shows there is a path to all $i\in\{N+1,\dots,2N\}\setminus \{N+k\}$ and the fact that $\beta_{ii}>0$ implies that there is a path to all $i\in\{1,\dots,N\}$. Then $P$ is non singular by Lemma \ref{lem:wcdd} and as the real parts of all the eigenvalues are non-positive, $\alpha(P)<0$. }
To see that $\alpha(P)=0$ { when $\gamma_{i}=1,~\forall i$} consider the vector $1_{2N}$. As $P1_{2N}=0_{2N}$, $1_{2N}$ is an eigenvector with a corresponding zero eigenvalue. As the real parts of the eigenvalues of $P$ are non-positive, 
$\alpha(P)$ must be zero.

\end{proof}
{ \noindent In the following proof we consider an upper bound of \eqref{eq:Pmat}: \begin{equation}P=\begin{bmatrix}
-\bar{B} & \bar{B} \\
W & -(\mathcal{L}_{o}+W)
\end{bmatrix},\end{equation} 
}
\begin{proof}[Proof of Theorem \ref{thm:0global}]
As the opinion graph is strongly connected {and $w_{i}>0,\forall i$ }then the matrix $P$ is irreducible and as such the eigenvalue $0$ is simple \cite{berman1994nonnegative}. By Lemma \ref{lem:posdiag}, there exists a positive diagonal matrix $Q$ such that $P^{T}Q+QP$ is negative semidefinite. Consider the Lyapunov function $V(z)=z^{T}Qz$. Then by Lemma \ref{lem:Pmat}
\vspace{-1ex}
\begin{align} \label{eq:upper}
\dot{V}(z)&=\dot{z}^{T}Qz+z^{T}Q\dot{z}\\
&\leq z^{T}(P^{T}Q+QP)z \leq 0.\\
\end{align}
\vspace{-5ex}

\noindent
In what follows we will argue that $\dot{V}(z)<0, ~\forall z\neq 0_{2N}$, allowing the use of $V(z)$ to show stability by Lyapunov's direct method. In the case that the upper bound $z^{T}(P^{T}Q+QP)z<0$, $\dot{V}(z)<0$ follows trivially. However when the upper bound $z^{T}(P^{T}Q+QP)z=0$, more analysis is required to show that $\dot{V}(z)<0$. Since $Q$ is a positive diagonal matrix, $z^{T}(P^{T}Q+QP)z=0$ is achieved when $z$ equals the right eigenvector of $P$ associated with zero, i.e. when $z\in\mathrm{span}(1_{2N})$.  

Consider the dynamics for $x_{i}$ when $x_{i}=o_{i}\neq 0$ and $x_{i}=o_{i}\neq 1$ then by the condition on the adoption parameters
		\begin{align}
		\dot{x}_i&=-\delta_i x_i (1-o_i) + (1-x_i)o_i\left( \sum_{\mcal{N}_i^A}\beta_{ij} x_j +\beta_{ii} \right) \\
		&=-\delta_i x_i(1-x_i) + (1-x_i)x_i\left( \sum_{\mcal{N}_i^A}\beta_{ij} x_j +\beta_{ii} \right) \\
		&=x_i (1-x_i)\left(-\delta_i + \sum_{\mcal{N}_i^A}\beta_{ij} x_j +\beta_{ii} \right) \\
		&\leq x_i (1-x_i)\left(-\delta_i + \Omega_i{(\tau)} \right) \\
		&<0.
		\end{align}
        
Considering the opinion dynamic, with the condition that $o_{i}=o_{j} \ \forall i,j\in \{1,\dots,N\}$ because $z\in\mathrm{span}(1_{2N})$ then  
        \begin{align}
        \dot{o}_{i}&=\sum_{j\in\mcal{N}_i^O} (o_j-o_i)+w_{i}^{x}\left({\gamma_{i}x_{i}}-o_{i}\right)\\
        &{\leq}        ~w_{i}^{x}\left(x_{i}-x_{i}\right)=0.\\
        \end{align}
        \vspace{-5ex}
        
\noindent Then separating $Q$ into  $Q=\begin{bmatrix} Q_{1}&0\\0&Q_{2}\end{bmatrix}$
 we have that
\begin{align}
        \dot{V}&=x^{T}Q_{1}\dot{x}+o^{T}Q_{2}\dot{o}\\
        &{\leq}~x^{T}Q_{1}\dot{x} <0
\end{align}
\vspace{-4ex}

\noindent
where negativity holds as $Q_{1}$ is a positive diagonal matrix.

Having shown when $z^{T}(P^{T}Q+QP)z=0$ that $\dot{V}<0$, we now consider the two remaining cases, that $z\notin\mathrm{span}(1_{2N})$ and $z=0_{2N}$. In the considered domain, $z^{T}(P^{T}Q+QP)z<0$ when $z\notin\mathrm{span}(1_{2N})$, which implies that $\dot{V}<0$. When $z=0_{2N}$, the Lyapunov function satisfies $V(0_{2N})=0$ and $\dot{V}(0_{2N})=0$. 
 Then the form of the Lyapunov function shows that $V(z)>0, \ z\neq0_{2N}$ on the considered domain. The above argument shows that $\dot{V}(z)<0, \ z\neq0_{2N}$. These together show the stability of $0_{2N}$ via Lyapunov's direct method.
\end{proof}

Next we show the stability results of $z^\ast=1_{2N}$.

 \begin{lemma}
    The equilibrium point $z^\ast=1_{2N}$ is locally stable if $\forall i, \ \ \Omega_i{(1)} > \delta_{i},~ { \gamma_{i}=1}$.
    \end{lemma}
    \begin{proof}
{ As $\gamma_{i}=1,~\forall i$, $z^\ast=1_{2N}$ is an equilbrium point.} ~Consider the Jacobian matrix at the equilibrium point $z^\ast=1_{2N}$ 
\begin{equation}\small
J(z^*) = \left[
\begin{array}{c|c}
\displaystyle \text{diag}\left(-\Omega_i{(1)}  \right) & \displaystyle \text{diag}\left(  \delta_{i} \right)  \\
\hline 
W & -(\mathcal{L}_{O}+W) 
\end{array}
\right].\label{eq:sjab2}\normalsize
\end{equation}

{ Similarly to the case of the equilibrium at $0_{2N}$, the condition $\Omega_i{(1)} > \delta_{i}, \; \forall i$ and Assumptions \ref{ass:w} and \ref{ass:og} show that $J(z^*)$ is WCDD. Since the diagonal elements are negative, the Jacobian is Hurwitz, by 
the Gershgorin Disc Theorem, which implies }
 $z=1_{2N}$ is a locally stable equilibrium. 
    \end{proof}
 \begin{proof}[Proof of Theorem \ref{th:as1}]
    To show asymptotic stability of $1_{2N}$, consider the change of variables $\hat{x}_{i}=1-x_{i}$ and $\hat{o}_{i}=1-o_{i}$. Then $\dot{\hat{x}}_{i}=-\dot{x}_{i}$ and $\dot{\hat{o}}_{i}=-\dot{o}_{i}$. It follows that:
    \vspace{-1ex}
\begin{align}
\dot{\hat{x}}_{i}&=\delta_i x_i (1-o_i) - (1-x_i)o_i\left( \sum_{\mcal{N}_i^A}\beta_{ij} x_j +\beta_{ii} \right) \\
&=\delta_i (1-\hat{x}_i)\hat{o}_i - \hat{x}_i(1-\hat{o}_i)\left( \sum_{\mcal{N}_i^A}\beta_{ij} (1-\hat{x}_j) +\beta_{ii} \right)
\end{align}
\vspace{-5ex}

\noindent
and
\begin{align}
\dot{\hat{o}}_{i}&=\sum_{\mcal{N}_i^O} (o_i-o_j)+w_{i}^{x}\left(o_{i}-{\gamma_{i}}x_{i}\right)\\ 
&=\sum_{\mcal{N}_i^O} (\hat{o}_j-\hat{o}_i)+w_{i}^{x}\left(\hat{x}_{i}-\hat{o}_{i}\right)\\ 
\end{align}
\vspace{-6ex}

\noindent
{as $\gamma_{i}=1,\forall i$}. 
~Consider the $\hat{x}_{i}$ dynamic: 
\vspace{-2ex}
\begin{align}
\dot{\hat{x}}_{i}&=\delta_i (1-\hat{x}_i)\hat{o}_i - \hat{x}_i(1-\hat{o}_i)\left( \sum_{\mcal{N}_i^A}\beta_{ij} (1-\hat{x}_j) +\beta_{ii} \right)\\
&\leq\delta_i (1-\hat{x}_i)\hat{o}_i - \hat{x}_i(1-\hat{o}_i)\left({\Omega_{i}(\tau)} \right)\\
&< \delta_i (1-\hat{x}_i)\hat{o}_i - \hat{x}_i(1-\hat{o}_i)\left(\delta_{i} \right)\\
&=\delta_{i}(\hat{o}_{i}-\hat{x}_{i}).
\end{align}
Then the matrix 
\vspace{-1ex}
\begin{equation}\hat{P}=\begin{bmatrix}
-D & D \\
W & -(\mathcal{L}_{o}+W)
\end{bmatrix}, \label{eq:Phatmat}\end{equation}
where $D=\mathrm{diag}\left(\delta_{i}\right)$, satisfies 

\noindent
Similar to the proof of Theorem \ref{thm:0global}, if the opinion graph is strongly connected {and $w_{i}^{x}>0, \forall i$} then $\hat{P}$ is irreducible and by Lemma \ref{lem:posdiag}, there exists a $\hat{Q}$ that renders $\hat{P}'\hat{Q}+\hat{Q}\hat{P}$ negative semi-definite. Then as in Theorem \ref{thm:0global}, one can use $\hat{z}^{T}\hat{Q}\hat{z}$ as a Lyapunov function to show stability of $1_{2N}$.
\end{proof}
The characterization of the unstable equilibrium follows.
    \begin{proof}[Proof of Lemma \ref{lem:eqast}]
Consider the dynamic in $x_{i}$ at the point $z^{\ast}$ under the assumption that $\delta_{i}=\sum_{\mcal{N}_i^A} \beta_{ij} x^{\ast}_{j} +\beta_{ii}$:
\vspace{-2ex}
\begin{align}
\dot{x}_{i}&=-\delta_{i}(1-o^{\ast}_{i})x^{\ast}_{i}+(1-x^{\ast}_{i})o^{\ast}_{i}\left(\sum_{\mcal{N}_i^A} \beta_{ij} x^{\ast}_{j} +\beta_{ii}\right)\\
&=-\delta_{i}(1-o^{\ast}_{i})x^{\ast}_{i}+(1-x^{\ast}_{i})o^{\ast}_{i}\delta_{i}\\
&=\delta_{i}(o^{\ast}_{i}-x^{\ast}_{i})
\end{align}
then $\dot{x}_{i}=0$ if $o^{\ast}_{i}=x_{i}^{\ast}$. Substituting in the other conditions: 
\begin{align}
\dot{o}_{i}&=\sum_{j\in\mathcal{N}_{o}^{i}} (o^{\ast}_{j}-o^{\ast}_{i})+({\gamma_{i}}x^{\ast}_{i}-o^{\ast}_{i})\\
&=\sum_{j\in\mathcal{N}_{o}^{i}} (x^{\ast}_{j}-x^{\ast}_{i})+(x^{\ast}_{i}-x^{\ast}_{i})\\
&=\sum_{j\in\mathcal{N}_{o}^{i}} (x^{\ast}_{j}-x^{\ast}_{i})\\
\end{align}
\vspace{-5ex}

\noindent
Then if $\sum_{j\in\mathcal{N}_{o}^{i}} (x^{\ast}_{j}-x^{\ast}_{i})=0$, $\dot{o}_{i}=0.$ If these conditions hold for all $i$ then $z^{\ast}$ is an equilibrium point.
\end{proof}
\begin{proof}[Proof of Theorem \ref{th:unstab}]
Consider the Jacobian at the equilibrium point $z^{\ast}$, the properties of which are described in Lemma \ref{lem:eqast}. The derivatives of the adoption dynamic at $z^{\ast}$ follow
\begin{align}
	\frac{\partial{f_i}}{\partial x_i} &= -\delta_i \\
	\frac{\partial{f_i}}{\partial x_j} &= \begin{cases} (1-x^{\ast}_i)x^{\ast}_i\beta_{ij} \ &\text{if} \ j\in\mcal{N}_i^A, j\neq i \\ 0 &\text{if} \ j\notin \mcal{N}_i^A \cup \{i\} \end{cases} \\
	\frac{\partial{f_i}}{\partial o_i} &= \delta_{i} \\
	\frac{\partial{f_i}}{\partial o_j} &=  0 \ \forall j \neq i.
\end{align}
The Jacobian can be written as 
\begin{equation}
J(z^{\ast}) = \left[
\begin{array}{c|c}
\displaystyle\frac{\partial f}{\partial x} & D \\
\hline
W & -(\mathcal{L}_{o}+W)
\end{array}
\right]
\end{equation}
where $D=\mathrm{diag}(\delta_{i})$. If the opinion graph $\mathcal{G}_{O}$ is strongly connected, the Jacobian is irreducible, which allows the application of Lemma \ref{lem:metzeig}. Consider a vector $y$ of the form $y=\begin{bmatrix} \alpha_{1} &
\alpha_{2} &
\dots &
\alpha_{N} &
{1}_{N}^{T}
\end{bmatrix}^{T}$
where $\alpha_{i}=1+\epsilon_{i}$ and $0<\epsilon_{i}<\frac{\sum_{\mcal{N}_{i}^{A}}\beta_{ij}(1-x_{i}^{\ast})x_{i}^{\ast}}{\delta_{i}}.$
Then consider the matrix product $J(z^{\ast})y$. The first $N$ entries of $J(z^{\ast})y$ follow:
\begin{align}
&-\delta_{i}\alpha_{i}+\alpha_{i}\left(\sum_{\mcal{N}_{i}^{A}}\beta_{ij}(1-x_{i}^{\ast})x_{i}^{\ast}\right)+\delta_{i}\\
&>-\delta_{i}\alpha_{i}+\left(\sum_{\mcal{N}_{i}^{A}}\beta_{ij}(1-x_{i}^{\ast})x_{i}^{\ast}\right)+\delta_{i}\\
&>-\delta_{i}\epsilon_{i}+\left(\sum_{\mcal{N}_{i}^{A}}\beta_{ij}(1-x_{i}^{\ast})x_{i}^{\ast}\right) \\
&>-\left(\sum_{\mcal{N}_{i}^{A}}\beta_{ij}(1-x_{i}^{\ast})x_{i}^{\ast}\right) +\left(\sum_{\mcal{N}_{i}^{A}}\beta_{ij}(1-x_{i}^{\ast})x_{i}^{\ast}\right)\\
&=0.
\end{align}
The last $N$ entries of $J(z^{\ast})y$ follow
\begin{align}
\alpha_{i}w_{i}^{x}-d_{i}^{O}-w_{i}^{x} + \sum_{\mcal{N}_i^O} 1 
&=(\alpha_{i}-1)w_{i}^{x}\\
&=\epsilon_{i}w_{i}^{x}\\
&>0.
\end{align}
As $y$ is elements-wise positive and the resulting vector $J(z^{\ast})y$ is element-wise positive, by Lemma \ref{lem:metzeig}, $\alpha(J(z^{\ast}))>0$ and the equilibrium point is unstable. 
\end{proof}

{ Finally we show stability of the adoption dynamic coupled with the bounded confidence opinion dynamic model shown in \eqref{eq:ob}. The non-positivity of the eigenvalues of the matrix $P$ from \eqref{eq:Pmat} and $\hat{P}$ from \eqref{eq:Phatmat}, which are used to characterize the stability of the equilibria of the coupled adoption behavior, does not depend on the structure of the opinion graph. As such the results for the stability of the equilibria of the coupled adoption model can be extended to the bounded confidence opinion dynamic model if the matrix associated with the upper bound is shown to be negative semidefinite.} Proving asymptotic stability requires the following definitions from the study of switched systems, which follow~\cite{liberzon2012switching}. Consider a family of systems, with some index set $\mathcal{P}$,
\begin{equation}\label{eq:fam}
\dot{x}=f_{p}(x) \ \ p\in\mathcal{P}
\end{equation}
which has a switching signal $\sigma(t):[0,\infty)\rightarrow\mathcal{P}$ which determines the switches between systems. This gives rise to a switched system,  
\vspace{-1.5ex}
\begin{equation}\label{eq:switch}
\dot{x}=f_{\sigma}(x).
\end{equation}
\begin{definition}
A switched system is uniformly asymptotically stable if it is asymptotically stable for all switching signals.  
\end{definition}
\begin{definition}\label{def:clf}
A positive definite $\mathcal{C}^{1}$ function $V$ is a common Lyapunov function for the family of systems in \eqref{eq:fam} if there is a positive definite continuous function $W$ such that
\begin{equation}
\frac{\partial V}{\partial t} f_{p}(x)\leq -W(x) \ \ \forall x\neq0, \ \ \forall p\in\mathcal{P}
\end{equation}
or equivalently if $\mathcal{P}$ is compact and 
\begin{equation}
\frac{\partial V}{\partial t} f_{p}(x) <0 \ \ \forall x\neq0, \ \ \forall p\in\mathcal{P}.
\end{equation}
\end{definition}
\begin{lemma}[Theorem 2.1 from \cite{liberzon2012switching}]\label{th:unistab}
If all systems in the family in \eqref{eq:fam} share a common Lyapunov function, then the switched system in \eqref{eq:switch} is uniformly asymptotically stable.  
\end{lemma}
\begin{proof}[Proof of \ref{th:as0bc}]
Consider the finite collection of opinion graph topologies $\hat{\mathcal{G}}_{o}=\{\mathcal{G}_{o}^{1},\mathcal{G}_{o}^{2},\dots,\mathcal{G}_{o}^{s}\}$ which the bounded confidence model can switch between. The original opinion graph $\mathcal{G}_{o}\in\hat{\mathcal{G}}_{o}$, and also the empty opinion graph $\mathcal{G}_{o}^{\emptyset}\in\hat{\mathcal{G}}_{o}$. Consider the graph $\mathcal{G}_{o}^{i}$ which consists of $k$ connected subgraphs. Then under the opinion dynamic on $\mathcal{G}_{o}^{i}$ and if $\delta_{i}>\Omega_i{(1)} \ \forall i$ the dynamics follow 
\vspace{-1.5ex}
\begin{equation}
\dot{z}\leq P_{i}z
\end{equation}
\vspace{-5ex}

\noindent
where 
\vspace{-1.5ex}
\begin{equation}
P_{i}=\begin{bmatrix}
-\bar{B} & \bar{B} \\
W & -\left(\begin{bmatrix}
\mathcal{L}_{o}^{1} & 0 & \dots & 0 \\
0&\mathcal{L}_{o}^{2} & \dots & 0 \\
0&0&\ddots&0 \\
0&0&\dots&\mathcal{L}_{o}^{k}
\end{bmatrix}+W\right)
\end{bmatrix}, 
\end{equation}
though a permutation may be required to put the opinion dynamics into this form. 
The matrix $P_{i}$ is negative semidefinite as it is diagonally dominant with negative diagonal elements and symmetric by 
assumption. 
Under $\mathcal{G}_{o}^{\emptyset}$ and if $\delta_{i}>\Omega_i{(1)} \ \forall i$ the dynamics follow 
\vspace{-1.5ex}
\begin{equation}
\dot{z}\leq \begin{bmatrix}
-\bar{B} & \bar{B} \\
W & -W
\end{bmatrix}z=P_{\emptyset}z.
\end{equation}

The matrix $P_{\emptyset}$ is also negative semidefinite as it is diagonally dominant with negative diagonal elements and symmetric by 
assumption.
Then under any opinion graph $\mathcal{G}^{j}\in\hat{\mathcal{G}}_{o}$ 
the function $V(z)=\frac{1}{2}z^{T}z$ satisfies 
\vspace{-1ex}
\begin{align}
\dot{V}&\leq z^{T}P_{j}z \\&\leq 0.
\end{align}
Unlike the case of Theorem \ref{thm:0global}, there are now multiple eigenvectors associated with the zero eigenvalue, specifically one for each connected component of the subgraph. In the extreme case of the empty opinion graph $\mathcal{G}_{o}^{\emptyset}$, if node $i$ satisfies $x_{i}=o_{i}=1$ then it will stay there indefinitely independent of the behavior of the other nodes in the system as this is an equilibrium point for $x_{i}$ and $o_{i}$. Now if $x_{i}=o_{i}=c$ for $0<c<1$ one can use a similar argument to Theorem \ref{thm:0global} to show that the chosen Lyapunov function has $\dot{V}<0$.

Therefore, the domain over which stability is considered has been modified to $[0,1)^{2N}$ to prevent such occurrences. This domain excludes the eigenvectors of $P_{j}$ associated with the zero eigenvalue that are also equilibrium points of the system. Therefore on the domain $[0,1)^{2N}$, $\dot{V}(z)<0 \ \forall z\neq 0_{2N}.$  As the set of possible graph topologies is finite, $V(x)=\frac{1}{2}z^{T}z$ serves as a common Lyapunov function by Definition \ref{def:clf} and can be used to show that the system is uniformly asymptotically stable by Lemma \ref{th:unistab}. 
\end{proof}
The theorem for the stability of the hit equilibrium point $z^{\ast}=1_{2N}$ is presented without proof as the proof follows similarly to the proof of Theorem \ref{th:as0bc}.
 \bibliographystyle{unsrt}
\bibliography{bib}

\end{document}